\documentclass[a4paper,reqno]{amsart}
\usepackage{amsmath,amssymb,mathrsfs}
\usepackage{txfonts,bm,pifont}
\usepackage{cite}
\usepackage{float}
\usepackage{stackrel}
\usepackage{subcaption}
\usepackage{graphicx,xcolor}
\usepackage{comment}
\usepackage{enumitem}
\usepackage{longtable}
\usepackage{hyperref}

\newtheorem{theorem}{Theorem}[section]

\newtheorem{corollary}[theorem]{Corollary}
\newtheorem{lemma}[theorem]{Lemma}
\newtheorem{remark}[theorem]{Remark}
\theoremstyle{definition}

\numberwithin{equation}{section}
\numberwithin{figure}{section}
\numberwithin{table}{section}
\newcommand{\wutilde}[1]{\vrule depth 0pt width 0pt%
{\raise0.8pt\hbox{$\smash{{\mathop{#1} \limits_{\displaystyle\widetilde{}}}}$}}}
\newcommand{\wuhat}[1]{\vrule depth 0pt width 0pt%
{\raise0.6pt\hbox{$\smash{{\mathop{#1} \limits_{\displaystyle\widehat{}}}}$}}}

\newcommand{\al}{\alpha}
\newcommand{\be}{\beta}

\newcommand{\ga}{\gamma}

\newcommand{\la}{\lambda}

\newcommand{\PDE}{P$\Delta$E}
\newcommand{\bbZ}{\mathbb{Z}}

\newcommand{\bbC}{\mathbb{C}}

\newcommand{\set}[2]{\left\{\left. #1 ~\right|~ #2 \right\}}

\newcommand{\ou}{\overline{u}}
\newcommand{\wu}{\widetilde{u}}
\newcommand{\wou}{\widetilde{\ou}}
\newcommand{\ov}{\overline{v}}
\newcommand{\wv}{\widetilde{v}}
\newcommand{\wov}{\widetilde{\ov}}
\newcommand{\ox}{\overline{x}}
\newcommand{\wx}{\widetilde{x}}
\newcommand{\wox}{\widetilde{\ox}}
\newcommand{\oX}{\overline{X}}
\newcommand{\wX}{\widetilde{X}}
\newcommand{\woX}{\widetilde{\oX}}
\newcommand{\otau}{\overline{\tau}}
\newcommand{\wtau}{\widetilde{\tau}}
\newcommand{\wotau}{\widetilde{\otau}}
\newcommand{\oomega}{\overline{\omega}}
\newcommand{\womega}{\widetilde{\omega}}
\newcommand{\woomega}{\widetilde{\oomega}}
\newcommand{\hu}{\widehat{u}}
\newcommand{\hou}{\widehat{\ou}}
\newcommand{\hwu}{\widehat{\wu}}
\newcommand{\hwou}{\widehat{\wou}}

\newcommand{\hx}{\widehat{x}}
\newcommand{\hox}{\widehat{\ox}}
\newcommand{\hwx}{\widehat{\wx}}

\newcommand{\htau}{\widehat{\tau}}
\newcommand{\hotau}{\widehat{\otau}}
\newcommand{\hwtau}{\widehat{\wtau}}
\newcommand{\hwotau}{\widehat{\wotau}}
\newcommand{\homega}{\widehat{\omega}}
\newcommand{\hoomega}{\widehat{\oomega}}
\newcommand{\hwomega}{\widehat{\womega}}
\newcommand{\hwoomega}{\widehat{\woomega}}
\makeatletter
\long\def\@makecaption#1#2{
 \vskip 10pt
 \setbox\@tempboxa\hbox{#1. #2}
 \ifdim \wd\@tempboxa >\hsize #1. #2\par \else \hbox
to\hsize{\hfil\box\@tempboxa\hfil}
 \fi}
\makeatother
\makeatletter
\@namedef{subjclassname@2020}{%
  \textup{2020} Mathematics Subject Classification}
\makeatother
\begin{document}
\allowdisplaybreaks

\title[]{Consistency around a cube property of Hirota's discrete KdV equation and the lattice sine-Gordon equation}
\author{Nobutaka Nakazono}
\address{Institute of Engineering, Tokyo University of Agriculture and Technology, 2-24-16 Nakacho Koganei, Tokyo 184-8588, Japan.}
\email{nakazono@go.tuat.ac.jp}
\begin{abstract}
It has been unknown whether Hirota's discrete Korteweg-de Vries equation and the lattice sine-Gordon equation have the consistency around a cube (CAC) property.
In this paper, we show that they have the CAC property.
Moreover, we also show that they can be extended to systems on the 3-dimensional integer lattice.
\end{abstract}

\maketitle

\section{Introduction}\label{Introduction}
In this paper, we focus on the following two integrable partial difference equations (\PDE s).
One is Hirota's discrete Korteweg-de Vries (dKdV) equation \cite{CNP1991:MR1111648,TGR2001:zbMATH01560028,kajiwara2008bilinearization}:
\begin{equation}\label{eqn:dKdV}
 u_{l+1,m+1}-u_{l,m}
 =\dfrac{\be_{m+1}-\al_l}{u_{l,m+1}}-\dfrac{\be_m-\al_{l+1}}{u_{l+1,m}},
\end{equation}
where $(l,m)\in\bbZ^2$ is a lattice parameter, $u_{l,m}\in\bbC$ is its function and $\al_l,\be_m\in\bbC$ are parameters that depend only on $l$ and $m$, respectively,
and the other is the lattice sine-Gordon (lsG) equation\cite{KN2018:zbMATH07031158}:
\begin{equation}\label{eqn:lsG}
 \dfrac{x_{l+1,m+1}}{x_{l,m}}
 =\left(\dfrac{p_{l+1}-q_mx_{l+1,m}}{q_m-p_{l+1}x_{l+1,m}}\right)
 \left(\dfrac{q_{m+1}-p_lx_{l,m+1}}{p_l-q_{m+1}x_{l,m+1}}\right),
\end{equation}
where $(l,m)\in\bbZ^2$ is a lattice parameter, $x_{l,m}\in\bbC$ is its function and $p_l,q_m\in\bbC$ are parameters that depend only on $l$ and $m$, respectively.

The {\PDE} \eqref{eqn:dKdV} is an integrable non-autonomous discrete version of the Korteweg-de Vries (KdV) equation\cite{KDV1895:zbMATH02679684}:
\begin{equation}\label{eqn:KdV}
 w_t+6ww_x+w_{xxx}=0,
\end{equation}
where $(t,x)\in\bbC^2$ and $w=w(t,x)\in\bbC$,
which is known as a mathematical model of waves on shallow water surfaces.
On the other hand, the {\PDE} \eqref{eqn:lsG} is an integrable non-autonomous discrete version of the sine-Gordon equation:
\begin{equation}\label{eqn:sg}
 \phi_{tt}-\phi_{xx}+\sin{\phi}=0,
\end{equation}
where $(t,x)\in\bbC^2$ and $\phi=\phi(t,x)\in\bbC$,
which is known as a motion equation of a row of pendulums hanging from a rod and being coupled by torsion springs.

The consistency around a cube (CAC) property\cite{NQC1983:MR719638,NCWQ1984:MR763123,QNCL1984:MR761644,NS1998:zbMATH01844203,NW2001:MR1869690} is known as an integrability of quad-equations.
(For the CAC property and the definition of quad-equation, see Appendix \ref{subsection:CAC_def}.)
As classifications of quad-equations by using the CAC property, the lists of equations by Adler-Bobenko-Suris (ABS)\cite{ABS2003:MR1962121,ABS2009:MR2503862}, Boll\cite{BollR2011:MR2846098,BollR2012:MR3010833} and Hietarinta \cite{HietarintaJ2019:zbMATH07053246} are known.
These classifications have been done under the following additional conditions:
\begin{enumerate}
\item[\cite{ABS2003:MR1962121}:]
~\\[-1.2em]
All face-equations on the same cube are similar equations differing only by the parameter values and are invariant under the group $D_4$ of the square symmetries.
Moreover, the CAC cube also has the tetrahedron property.
\item[\cite{ABS2009:MR2503862,BollR2011:MR2846098,BollR2012:MR3010833}:]
The face-equations on the same cube are allowed to be different, and no assumption is made about symmetry.
However, the tetrahedron property is still required.
\item[\cite{HietarintaJ2019:zbMATH07053246}:]
No assumptions are made for symmetry or the tetrahedron property, but all face-equations are assumed to be homogeneous quadratic.
The face-equations on the same cube allow different equations on the three orthogonal planes but are the same equations on parallel planes.
\end{enumerate}
Note here that a CAC-cube means a cube with the six quad-equations on its faces (face-equations) which has the CAC property.
However, the \PDE s \eqref{eqn:dKdV} and \eqref{eqn:lsG} are not included in these lists, and the structures of their consistency have not been reported.

Recently, it has been shown that the \PDE s \eqref{eqn:dKdV} and \eqref{eqn:lsG} have the consistency around a broken cube (CABC) property\cite{JN2021:zbMATH07476241,nakazono2022properties}.
(For the CABC property, see Appendix \ref{subsection:CABC_def}.)
By considering the B\"acklund transformation (BT)
\begin{subequations}\label{eqns:dKdV_CABC_1}
\begin{align}
 &u_{l,m}-v_{l+1,m+1}-\dfrac{\be_m-\al_{l+1}}{u_{l+1,m}}+\dfrac{\ga-\al_l}{v_{l,m+1}}=0,
 \label{eqn:dKdV_CABC_1_S}\\
 &(u_{l,m}-v_{l,m})\left(\dfrac{\be_m-\al_l}{u_{l,m}}+v_{l,m+1}\right)-\be_m+\ga=0,
 \label{eqn:dKdV_CABC_1_B}\\
 &\dfrac{\be_m-\al_l}{u_{l,m}}-\dfrac{\ga-\al_l}{v_{l,m}}-u_{l+1,m}+v_{l+1,m}=0,
 \label{eqn:dKdV_CABC_1_C}
\end{align}
\end{subequations}
 from the {\PDE} \eqref{eqn:dKdV} to the multi-quadratic equation
\begin{equation}\label{eqn:multiquadratic}
\begin{split}
 &\be_m\left(v_{l,m}\,v_{l+1,m}-v_{l,m+1}\,v_{l+1,m+1}\right)^2\\
 &\quad +\bigl(v_{l,m}-v_{l+1,m+1}\bigr)\bigl(v_{l+1,m}-v_{l,m+1}\bigr)\bigl(\ga-v_{l,m}\,v_{l+1,m}\bigr)\bigl(\ga-v_{l,m+1}\,v_{l+1,m+1}\bigr)\\
 &\quad +\bigl(\al_{l+1}v_{l,m}-\al_lv_{l+1,m+1}\bigr)\bigl(\al_lv_{l+1,m}-\al_{l+1}v_{l,m+1}\bigr)\\
 &\quad -(\al_l+\al_{l+1})\bigl(\ga\,v_{l,m}\,v_{l+1,m}+\ga\,v_{l,m+1}\,v_{l+1,m+1}-2\,v_{l,m}\,v_{l+1,m}\,v_{l,m+1}\,v_{l+1,m+1}\bigr)\\
 &\quad+\bigl(\al_{l+1}\,v_{l,m}\,v_{l,m+1}+\al_l\,v_{l+1,m}\,v_{l+1,m+1}\bigr)\bigl(2\ga - v_{l,m} v_{l+1,m}-v_{l,m+1} \,v_{l+1,m+1}\bigr) =0,
\end{split}
\end{equation}
where $\ga\in\bbC$ is a constant parameter, the CABC property of the {\PDE} \eqref{eqn:dKdV} was shown in \cite{JN2021:zbMATH07476241}.
Similarly, by considering the following BT from the {\PDE} \eqref{eqn:lsG} to Equation \eqref{eqn:multiquadratic}:
\begin{subequations}\label{eqns:lsG_CABC}
\begin{align}
 &\dfrac{v_{l,m+1}v_{l+1,m+1}}{p_lp_{l+1}}-\dfrac{x_{l,m}(p_{l+1}-q_m x_{l+1,m})}{p_{l+1}x_{l+1,m}-q_m}=0,
 \label{eqn:lsG_CABC_S}\\
 &\dfrac{v_{l,m}v_{l,m+1}}{{p_l}^2}-\dfrac{x_{l,m}(p_l-q_mx_{l,m})}{p_lx_{l,m}-q_m}=0,
 \label{eqn:lsG_CABC_B}\\
 &\dfrac{v_{l,m}v_{l+1,m}}{p_lp_{l+1}}-\dfrac{x_{l+1,m}(p_l-q_mx_{l,m})}{p_lx_{l,m}-q_m}=0,
 \label{eqn:lsG_CABC_C}
\end{align}
\end{subequations}
the CABC property of the {\PDE} \eqref{eqn:lsG} was shown in \cite{nakazono2022properties}.
In this case, the relation between the parameters $p_l$, $q_m$, $\al_l$, $\be_m$ and $\ga$ is given by
\begin{equation}\label{eqn:pq_albe_ga}
 p_l=\sqrt{\al_l-\ga},\quad
 q_m=\sqrt{\be_m-\ga}.
\end{equation}
Since having a CABC property does not imply not having a CAC property, the CABC property of the \PDE s \eqref{eqn:dKdV} and \eqref{eqn:lsG} was shown, but whether they have the CAC property or not is an open problem.
It is natural to consider that the CAC property can be shown using the CABC property because both are meant to show integrability and are the properties around a cube.
This is the motivation for this study.

\begin{remark}
The first condition in Equation \eqref{eqn:pq_albe_ga} means that given the complex parameters $\al_l$ and $\ga$, we determine one value of the complex parameter $p_l$ that satisfies 
\begin{equation}
 {p_l}^2=\al_l-\ga.
\end{equation}
The same is true for the second condition in Equation \eqref{eqn:pq_albe_ga}.
In addition, if the root symbol $\sqrt{~\,~}$ is used in the following, it means the same thing.
\end{remark}

In this paper, we introduce another CABC-type BT from the {\PDE} \eqref{eqn:dKdV} to Equation \eqref{eqn:multiquadratic}, which is different from the CABC-type BT \eqref{eqns:dKdV_CABC_1}. 
Then, using a total of three CABC-type BTs, we show the following results:
\begin{enumerate}
\item[(i)] 
derive two CAC-type BTs from the {\PDE} \eqref{eqn:dKdV} to the {\PDE} \eqref{eqn:lsG};
\item[(ii)]  
derive a CAC-type auto-B\"acklund transformation (auto-BT) of the {\PDE} \eqref{eqn:lsG};
\item[(iii)] 
construct Lax pairs of the \PDE s \eqref{eqn:dKdV} and \eqref{eqn:lsG};
\item[(iv)] 
derive integrable systems of \PDE s on the 3-dimensional integer lattice containing the \PDE s \eqref{eqn:dKdV} and \eqref{eqn:lsG}.
\end{enumerate}

The main contribution of this paper is not only to show that the systems \eqref{eqn:dKdV} and \eqref{eqn:lsG} have the CAC property but also to find examples of quad-equations having the CAC property that do not satisfy the conditions for classifying the list of equations by ABS, Boll and Hietarinta above.
Finding such examples leads to the study of making new lists of \PDE s that have the CAC property, i.e., it contributes to the development of the study of integrable systems.

\subsection{Notation and Terminology}\label{subsection:NT}
\begin{itemize}
\item 
For simplicity, we use the following shorthand notations:
\begin{equation}
 \overline{\rule{0em}{0.5em}\hspace{0.6em}}:\,l\to l+1,\quad
 \widetilde{}\,:\,m\to m+1,\quad
 \widehat{}\,:\,n\to n+1.
\end{equation}
For example, 
\begin{equation}
 \ou=u_{l+1,m},\quad
 \wu=u_{l,m+1},\quad
 \wou=u_{l+1,m+1},\quad
 \overline{\ou}=u_{l+2,m},\quad
 \widetilde{\wu}=u_{l,m+2}
\end{equation}
for $u=u_{l,m}$ and 
\begin{equation}
 \ou=u_{l+1,m,n},\quad
 \wu=u_{l,m+1,n},\quad
 \hu=u_{l,m,n+1},\quad
 \hwou=u_{l+1,m+1,n+1}
\end{equation}
for $u=u_{l,m,n}$, and the like.

\item
If there is a substitution notation (e.g., $l\to l+1$) in the subscript of an equation number, it means the equation with the substitution applied to all symbols in the corresponding equation.
For example, \eqref{eqn:dKdV}$_{l\to l+1}$ implies the following equation:
\[
 u_{l+2,m+1}-u_{l+1,m}
 =\dfrac{\be_{m+1}-\al_{l+1}}{u_{l+1,m+1}}-\dfrac{\be_m-\al_{l+2}}{u_{l+2,m}}.
\]

\item
The following terminologies are used in this paper.
\begin{enumerate}
\item 
A multivariate polynomial with complex coefficients that is affine linear with respect to each variable is called a {\it multilinear polynomial}.
For example, the general form of a multilinear polynomial in four variables is given by
\begin{align}
 &A_1xyzw+A_2xyz+A_3xyw+A_4xzw+A_5yzw+A_6xy+A_7xz+A_8xw\notag\\
 &\quad +A_9yz+A_{10}yw+A_{11}zw+A_{12}x+A_{13}y+A_{14}z+A_{15}w+A_{16},
\end{align}
where $A_i\in\bbC$ are parameters.
\item 
Let $Q(x,y,z,w)$ be a multilinear polynomial and $x=f(y,z,w)$ be the solution of $Q(x,y,z,w)=0$.
We call $Q(x,y,z,w)$ an {\it irreducible multilinear polynomial} in four variables if the following holds:
\begin{equation}
 \dfrac{\partial}{\partial y} f(y,z,w)\neq0,\quad
 \dfrac{\partial}{\partial z} f(y,z,w)\neq0,\quad
 \dfrac{\partial}{\partial w} f(y,z,w)\neq0.
\end{equation}
\item
Consider a {\PDE} of the form 
\begin{equation}\label{eqn:general_quad_eqn}
 Q(u,\ou,\wu,\wou\,)=0,
\end{equation}
where $u=u_{l,m}$ and $Q$ is an irreducible multilinear polynomial. 
By putting the variable $u_{l,m}$ on the grid point $(l,m)$ of the integer lattice $\bbZ^2$,
the {\PDE} \eqref{eqn:general_quad_eqn} can be regarded as a relation on the quadrilateral given by
\[
 (l,m),\quad (l+1,m),\quad (l,m+1),\quad (l+1,m+1)
\]
as shown in Figure \ref{fig:quadrilateral}.
For this reason, a {\PDE} of the form \eqref{eqn:general_quad_eqn} is called a {\it quad-equation}.
Note that quad-equations are not necessarily autonomous but may contain parameters that evolve with $l$ or $m$.
\end{enumerate}
\end{itemize}

\begin{figure}[htbp]
\begin{center}
 \includegraphics[width=0.35\textwidth]{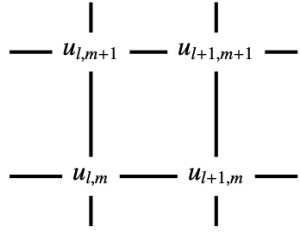}
\end{center}
\caption{A quadrilateral for the quad-equation \eqref{eqn:general_quad_eqn}.}
\label{fig:quadrilateral}
\end{figure}

\subsection{Plan of the paper}
This paper is organized as follows.
In \S \ref{section:dKdV_lsG_consistency}, using the CABC property of the dKdV equation \eqref{eqn:dKdV} and the lsG equation \eqref{eqn:lsG}, we show that they have the CAC property.
In \S \ref{section:3D}, we show how the \PDE s \eqref{eqn:dKdV} and \eqref{eqn:lsG} can be extended to systems of \PDE s on the 3-dimensional integer lattice.
Some concluding remarks are given in \S \ref{ConcludingRemarks}.
In Appendix \ref{section:CAC_CABC}, we recall the definitions of the CAC and CABC properties.
In Appendix \ref{section:Lax}, we list the Lax pairs of the \PDE s \eqref{eqn:dKdV} and \eqref{eqn:lsG} constructed in this study.
In Appendix \ref{section:proof_tau}, we give the proof of Theorem \ref{theorem:tau}.
\section{Consistencies of the dKdV equation and the lsG equation}\label{section:dKdV_lsG_consistency}
In this section, we give the new consistencies of the dKdV equation \eqref{eqn:dKdV} and the lsG equation \eqref{eqn:lsG}.
See Appendix \ref{section:CAC_CABC} for descriptions of the CABC and CAC properties we deal with here.

\subsection{A CABC-type BT from the {\PDE} \eqref{eqn:dKdV} to the multi-quadratic equation \eqref{eqn:multiquadratic}}\label{subsection:dKdV_CABC}
The following system of equations is a new CABC-type BT from the {\PDE} \eqref{eqn:dKdV} to Equation \eqref{eqn:multiquadratic}:
\begin{subequations}\label{eqns:dKdV_CABC_2}
\begin{align}
 &u-\wv-\dfrac{\be_m-\al_{l+1}}{\ou}+\dfrac{\ga-\al_{l+1}}{\wov}=0,
 \label{eqn:dKdV_CABC_2_S}\\
 &(u-\wv)\left(\dfrac{\be_m-\al_l}{u}+v\right)-\be_m+\ga=0,
 \label{eqn:dKdV_CABC_2_B}\\
 &\dfrac{\be_m-\al_l}{u}-\dfrac{\ga-\al_{l+1}}{\ov}-\ou+v=0,
 \label{eqn:dKdV_CABC_2_C}
\end{align}
\end{subequations}
where $u=u_{l,m}$ and $v=v_{l,m}$.
By direct calculation, we can verify that the system of equations \eqref{eqn:dKdV} and \eqref{eqns:dKdV_CABC_2} has the CABC and tetrahedron properties.
Indeed, the following relations can be derived from the CABC-system \eqref{eqn:dKdV} and \eqref{eqns:dKdV_CABC_2}:
\begin{equation}
 \dfrac{\al_{l+1}-\ga}{\wov}=\dfrac{\al_{l+1}-\be_m}{\ou}-\dfrac{(\be_m-\ga)u}{\al_l-\be_m-uv},\quad
 \wv=u+\dfrac{(\be_m-\ga)\ov}{\al_{l+1}-\ga-\ou\,\ov}.
\end{equation}
As shown in \cite{JN2021:zbMATH07476241,nakazono2022properties}, the Lax pair \eqref{eqn:Lax_phi} with \eqref{eqn:Lax_dKdV_CABC} of the {\PDE} \eqref{eqn:dKdV} can be constructed from the equations \eqref{eqn:dKdV_CABC_2_B} and \eqref{eqn:dKdV_CABC_2_C} by setting
\begin{equation}
 v_{l,m}=\dfrac{F_{l,m}}{G_{l,m}},\quad
 \Phi_{l,m}=\begin{pmatrix}F_{l,m}\\G_{l,m}\end{pmatrix}.
\end{equation}
\subsection{A CAC-type BT from the {\PDE} \eqref{eqn:dKdV} to the {\PDE} \eqref{eqn:lsG} (I)}\label{subsection:CABC_to_CAC_1}
In this subsection, using the CABC-type BTs \eqref{eqns:dKdV_CABC_1} and \eqref{eqns:lsG_CABC}, we obtain a CAC-type BT from the {\PDE} \eqref{eqn:dKdV} to the {\PDE} \eqref{eqn:lsG}.

Eliminating $\wv$ from the equations \eqref{eqn:dKdV_CABC_1_B} and \eqref{eqn:lsG_CABC_B}, we obtain 
\begin{equation}\label{eqn:v_ux_CABC_1}
 \Big((q_m-p_lx)v+p_lux\Big)
 \Big(({p_l}^2-{q_m}^2)v-p_lu(p_l-q_mx)\Big)
 =0.
\end{equation}
Let us assume that the following holds:
\begin{equation}\label{eqn:v_ux_CABC_2}
 v=\dfrac{p_lux}{p_lx-q_m},
\end{equation}
which satisfies Equation \eqref{eqn:v_ux_CABC_1}.
Eliminating $v$ from Equation \eqref{eqn:dKdV_CABC_1_B} by using Equation \eqref{eqn:v_ux_CABC_2}
and $\wv$ from Equation \eqref{eqn:dKdV_CABC_1_C} by using Equation \eqref{eqn:v_ux_CABC_2}$_{m\to m+1}$, 
we obtain
\begin{subequations}
\begin{align}
 &\wu u+\dfrac{(p_l-q_mx)(q_{m+1}-p_l\wx)}{\wx}=0,\\
 &\ou u+\dfrac{(p_l-q_mx)(q_m-p_{l+1}\ox)}{x}=0,
\end{align}
\end{subequations}
respectively.
The equations above, together with the \PDE s \eqref{eqn:dKdV} and \eqref{eqn:lsG}, lead to the following theorem.

\begin{theorem}\label{thm:CAC1}
The following system has the CAC property but does not have the tetrahedron property:
\begin{subequations}\label{eqns:dKdV_lsG_CAC_1}
\begin{align}
 &\wou-u=\dfrac{\be_{m+1}-\al_l}{\wu}-\dfrac{\be_m-\al_{l+1}}{\ou},
 \label{eqn:dKdV_lsG_CAC_1_A}\\
 &\wu u+\dfrac{(p_l-q_mx)(q_{m+1}-p_l\wx)}{\wx}=0,
 \label{eqn:dKdV_lsG_CAC_1_B}\\
 &\ou u+\dfrac{(p_l-q_mx)(q_m-p_{l+1}\ox)}{x}=0,
 \label{eqn:dKdV_lsG_CAC_1_C}\\
 &\dfrac{\,\wox\,}{x}
 =\left(\dfrac{p_{l+1}-q_m\ox}{q_m-p_{l+1}\ox}\right)
 \left(\dfrac{q_{m+1}-p_l\wx}{p_l-q_{m+1}\wx}\right),
 \label{eqn:dKdV_lsG_CAC_1_AA}
\end{align}
\end{subequations}
where $p_l$ and $q_m$ are given by Equation \eqref{eqn:pq_albe_ga}.
Since the equations \eqref{eqn:dKdV_lsG_CAC_1_A} and \eqref{eqn:dKdV_lsG_CAC_1_AA} are equal to the \PDE s \eqref{eqn:dKdV} and \eqref{eqn:lsG}, respectively, we can also say that the pair of equations \eqref{eqn:dKdV_lsG_CAC_1_B} and \eqref{eqn:dKdV_lsG_CAC_1_C} is a CAC-type BT from the {\PDE} \eqref{eqn:dKdV} to the {\PDE} \eqref{eqn:lsG}.
\end{theorem}
\begin{proof}
Consider the cube in Figure \ref{fig:cube_ux}.
In the three different ways (see Appendix \ref{subsection:CAC_def}), $\wox$ can be uniquely represented by the initial values $\{u,\ou,\wu,x\}$ as
\begin{equation}
 \wox=\dfrac{q_{m+1}\wu\,\Big((\al_{l+1}-\be_m)(p_l-q_mx)-q_m u\ou x\Big)}{p_{l+1}\ou\,\Big((\al_l-\be_{m+1})(p_l-q_mx)-p_l u\wu\Big)}.
\end{equation}
Since $\wox$ depends on $u$, the tetrahedron property does not hold.
Therefore, we have completed the proof.
\end{proof}

Theorem \ref{thm:CAC1} gives the following corollary.

\begin{corollary}
Both of the \PDE s \eqref{eqn:dKdV} and \eqref{eqn:lsG} have the CAC property.
\end{corollary}

\begin{figure}[htbp]
\begin{center}
 \includegraphics[width=0.45\textwidth]{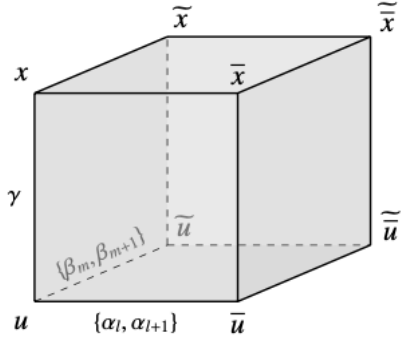}
\end{center}
\caption{
A cube for the CAC-systems \eqref{eqns:dKdV_lsG_CAC_1} and \eqref{eqns:dKdV_lsG_CAC_2}.
The $u$- and $x$-variables are assigned at the bottom and top vertices, respectively.
The corresponding equations of a CAC-system are assigned to the faces, e.g.
the equations \eqref{eqn:dKdV_lsG_CAC_1_A} and \eqref{eqn:dKdV_lsG_CAC_1_AA} are assigned to the bottom and top faces, respectively.
Also, as a convention, the parameters in a CAC-system are placed on the line segments.
Thus, the parameters $\{\al_l,\al_{l+1}\}$ are assigned to the line segments connecting the variables that represent a shift in the $l$-direction (e.g., $u$ and $\ou$) since $\{\al_l,\al_{l+1}\}$ depend on only $l$.
Similarly, the parameters $\{\be_m,\be_{m+1}\}$ are assigned to the line segments connecting the variables that represent a shift in the $m$-direction.
On the other hand, the parameter $\ga$ is assigned to the line segments connecting the $u$- and $x$-varibales (e.g., $u$ and $x$) since $\ga$ is a parameter associated with the BT from the quad-equation given by $u$-variable to that given by $x$-variable.
For simplicity, only one parameter is shown for each in the figure, but they can also be placed on three other parallel line segments.
For example, $\ga$ can also be placed on the line segment connecting $\ou$ and $\ox$.
}
\label{fig:cube_ux}
\end{figure}

\begin{remark}
If we choose 
\begin{equation}
 v=\dfrac{p_lu(p_l-q_mx)}{{p_l}^2-{q_m}^2}
\end{equation}
as the solution of Equation \eqref{eqn:v_ux_CABC_1} instead of Equation \eqref{eqn:v_ux_CABC_2} and construct another system of \PDE s, then we can verify that the resulting system does not have CAC property.
\end{remark}

As shown in \cite{BS2002:MR1890049,HJN2016:MR3587455,NijhoffFW2002:MR1912127,WalkerAJ:thesis}, the Lax pair \eqref{eqn:Lax_phi} with \eqref{eqn:Lax_dKdV_CAC1} of the {\PDE} \eqref{eqn:dKdV} can be constructed from the equations \eqref{eqn:dKdV_lsG_CAC_1_B} and \eqref{eqn:dKdV_lsG_CAC_1_C} by setting
\begin{equation}
 x_{l,m}=\dfrac{F_{l,m}}{G_{l,m}},\quad
 \Phi_{l,m}=\begin{pmatrix}F_{l,m}\\G_{l,m}\end{pmatrix}.
\end{equation}
 
\begin{remark}
Similar to the Lax pair of the {\PDE} \eqref{eqn:dKdV}, 
we can also construct a Lax pair of the {\PDE} \eqref{eqn:lsG} from the equations \eqref{eqn:dKdV_lsG_CAC_1_B} and \eqref{eqn:dKdV_lsG_CAC_1_C} by setting 
\begin{equation}
 u_{l,m}=\dfrac{F_{l,m}}{G_{l,m}},\quad
 \Phi_{l,m}=\begin{pmatrix}F_{l,m}\\G_{l,m}\end{pmatrix}.
\end{equation}
However, we still have not figured out how to introduce a spectral parameter in this case.
\end{remark}
 
\subsection{A CAC-type BT from the {\PDE} \eqref{eqn:dKdV} to the {\PDE} \eqref{eqn:lsG} (II)}\label{subsection:CABC_to_CAC_2}
In this subsection, using the CABC-type BTs \eqref{eqns:lsG_CABC} and \eqref{eqns:dKdV_CABC_2}, we obtain a CAC-type BT from the {\PDE} \eqref{eqn:dKdV} to the {\PDE} \eqref{eqn:lsG}.
Since the process for demonstrating the result is the same as that for the CABC-type BTs \eqref{eqns:dKdV_CABC_1} and \eqref{eqns:lsG_CABC} discussed in \S \ref{subsection:CABC_to_CAC_1}, we omit detailed arguments.

\begin{theorem}\label{thm:CAC2}
The following system has the CAC property but does not have the tetrahedron property:
\begin{subequations}\label{eqns:dKdV_lsG_CAC_2}
\begin{align}
 &\wou-u=\dfrac{\be_{m+1}-\al_l}{\wu}-\dfrac{\be_m-\al_{l+1}}{\ou},
 \label{eqn:dKdV_lsG_CAC_2_A}\\
 &\wu u+\dfrac{(q_m-p_lx)(p_l-q_{m+1}\wx)}{x}=0,
 \label{eqn:dKdV_lsG_CAC_2_B}\\
 &\ou u+\dfrac{(q_m-p_lx)(p_{l+1}-q_m\ox)}{\ox}=0,
 \label{eqn:dKdV_lsG_CAC_2_C}\\
 &\dfrac{\,\wox\,}{x}
 =\left(\dfrac{p_{l+1}-q_m\ox}{q_m-p_{l+1}\ox}\right)
 \left(\dfrac{q_{m+1}-p_l\wx}{p_l-q_{m+1}\wx}\right),
 \label{eqn:dKdV_lsG_CAC_2_AA}
\end{align}
\end{subequations}
where $p_l$ and $q_m$ are given by Equation \eqref{eqn:pq_albe_ga}.
Since the equations \eqref{eqn:dKdV_lsG_CAC_2_A} and \eqref{eqn:dKdV_lsG_CAC_2_AA} are equal to the \PDE s \eqref{eqn:dKdV} and \eqref{eqn:lsG}, respectively, we can also say that the pair of equations \eqref{eqn:dKdV_lsG_CAC_2_B} and \eqref{eqn:dKdV_lsG_CAC_2_C} is a CAC-type BT from the {\PDE} \eqref{eqn:dKdV} to the {\PDE} \eqref{eqn:lsG}.
\end{theorem}
\begin{proof}
Consider the cube in Figure \ref{fig:cube_ux}.
In the three different ways (see Appendix \ref{subsection:CAC_def}), $\wox$ can be uniquely represented by the initial values $\{u,\ou,\wu,x\}$ as
\begin{equation}
 \wox=\dfrac{p_{l+1}\ou\,\Big((\al_l-\be_{m+1})(q_m-p_lx)+p_l u\wu x\Big)}{q_{m+1}\wu\,\Big((\al_{l+1}-\be_m)(q_m-p_lx)+q_m u\ou\Big)}.
\end{equation}
Since $\wox$ depends on $u$, the tetrahedron property does not hold.
Therefore, we have completed the proof.
\end{proof}

Setting
\begin{equation}
 x_{l,m}=\dfrac{F_{l,m}}{G_{l,m}},\quad
 \Phi_{l,m}=\begin{pmatrix}F_{l,m}\\G_{l,m}\end{pmatrix},
\end{equation}
from the equations \eqref{eqn:dKdV_lsG_CAC_2_B} and \eqref{eqn:dKdV_lsG_CAC_2_C}, we obtain the Lax pair \eqref{eqn:Lax_phi} with \eqref{eqn:Lax_dKdV_CAC2} of the {\PDE} \eqref{eqn:dKdV}.

\begin{remark}
Similar to the Lax pair of the {\PDE} \eqref{eqn:dKdV}, 
we can also construct a Lax pair of the {\PDE} \eqref{eqn:lsG} from the equations \eqref{eqn:dKdV_lsG_CAC_2_B} and \eqref{eqn:dKdV_lsG_CAC_2_C} by setting 
\begin{equation}
 u_{l,m}=\dfrac{F_{l,m}}{G_{l,m}},\quad
 \Phi_{l,m}=\begin{pmatrix}F_{l,m}\\G_{l,m}\end{pmatrix}.
\end{equation}
However, in this case, how to introduce a spectral parameter is still unsolved at this point.
\end{remark}

\subsection{A CAC-type auto-BT of the {\PDE} \eqref{eqn:lsG}}\label{subsection:2CACs_to_CAC}
In this subsection, using two CAC-systems \eqref{eqns:dKdV_lsG_CAC_1} and \eqref{eqns:dKdV_lsG_CAC_2}, we obtain a CAC-system, which implies a CAC-type auto-BT of the {\PDE} \eqref{eqn:lsG}.

Here, let us use the following CAC-system:
\begin{subequations}\label{eqns:dKdV_lsG_CAC_1_XPQ}
\begin{align}
 &\wou-u=\dfrac{\be_{m+1}-\al_l}{\wu}-\dfrac{\be_m-\al_{l+1}}{\ou},
 \label{eqn:dKdV_lsG_CAC_1_A_XPQ}\\
 &\wu u+\dfrac{(P_l-Q_mX)(Q_{m+1}-P_l\wX)}{\wX}=0,
 \label{eqn:dKdV_lsG_CAC_1_B_XPQ}\\
 &\ou u+\dfrac{(P_l-Q_mX)(Q_m-P_{l+1}\oX)}{X}=0,
 \label{eqn:dKdV_lsG_CAC_1_C_XPQ}\\
 &\dfrac{\,\woX\,}{X}
 =\left(\dfrac{P_{l+1}-Q_m\oX}{Q_m-P_{l+1}\oX}\right)
 \left(\dfrac{Q_{m+1}-P_l\wX}{P_l-Q_{m+1}\wX}\right),
 \label{eqn:dKdV_lsG_CAC_1_AA_XPQ}
\end{align}
\end{subequations}
where $X=X_{l,m}$ and 
\begin{equation}\label{eqn:PQ_albe_gala}
 P_l=\sqrt{\al_l-\ga-\la},\quad
 Q_m=\sqrt{\be_m-\ga-\la}\,,
\end{equation}
instead of the CAC-system \eqref{eqns:dKdV_lsG_CAC_1} with the variable $x_{l,m}$ and the parameters $p_l$, $q_m$ in the system \eqref{eqns:dKdV_lsG_CAC_1} replaced by the variable $X_{l,m}$ and the parameters $P_l$, $Q_m$, respectively.
Eliminating $\wu u$ from the equations \eqref{eqn:dKdV_lsG_CAC_2_B} and \eqref{eqn:dKdV_lsG_CAC_1_B_XPQ}
and $\ou u$ from the equations \eqref{eqn:dKdV_lsG_CAC_2_C} and \eqref{eqn:dKdV_lsG_CAC_1_C_XPQ}, we obtain
\begin{subequations}
\begin{align}
 &\dfrac{(q_m-p_lx)(p_l-q_{m+1}\wx)}{x}=\dfrac{(P_l-Q_mX)(Q_{m+1}-P_l\wX)}{\wX},\\
 &\dfrac{(q_m-p_lx)(p_{l+1}-q_m\ox)}{\ox}=\dfrac{(P_l-Q_mX)(Q_m-P_{l+1}\oX)}{X},
\end{align}
\end{subequations}
respectively.
The equations above, together with the equations \eqref{eqn:dKdV_lsG_CAC_2_AA} and \eqref{eqn:dKdV_lsG_CAC_1_AA_XPQ}, lead to the following theorem.

\begin{theorem}\label{thm:CAC3}
The following system has the CAC property but does not have the tetrahedron property:
\begin{subequations}\label{eqns:lsG_CAC}
\begin{align}
 &\dfrac{\,\wox\,}{x}
 =\left(\dfrac{p_{l+1}-q_m\ox}{q_m-p_{l+1}\ox}\right)
 \left(\dfrac{q_{m+1}-p_l\wx}{p_l-q_{m+1}\wx}\right),
 \label{eqn:lsG_CAC_A}\\
 &\dfrac{(q_m-p_lx)(p_l-q_{m+1}\wx)}{x}=\dfrac{(P_l-Q_mX)(Q_{m+1}-P_l\wX)}{\wX},
 \label{eqn:lsG_CAC_B}\\
 &\dfrac{(q_m-p_lx)(p_{l+1}-q_m\ox)}{\ox}=\dfrac{(P_l-Q_mX)(Q_m-P_{l+1}\oX)}{X},
 \label{eqn:lsG_CAC_C}\\
 &\dfrac{\,\woX\,}{X}
 =\left(\dfrac{P_{l+1}-Q_m\oX}{Q_m-P_{l+1}\oX}\right)
 \left(\dfrac{Q_{m+1}-P_l\wX}{P_l-Q_{m+1}\wX}\right),
 \label{eqn:lsG_CAC_AA}
\end{align}
\end{subequations}
where
\begin{equation}\label{eqn:pq_PQ}
 P_l=\sqrt{{p_l}^2-\la},\quad
 Q_m=\sqrt{{q_m}^2-\la}\,.
\end{equation}
Since both of the equations \eqref{eqn:lsG_CAC_A} and \eqref{eqn:lsG_CAC_AA} are equal to the {\PDE} \eqref{eqn:lsG}, we can also say that the pair of equations \eqref{eqn:lsG_CAC_B} and \eqref{eqn:lsG_CAC_C} is a CAC-type auto-BT of the {\PDE} \eqref{eqn:lsG}.
\end{theorem}
\begin{proof}
Consider the cube in Figure \ref{fig:cube_xX}.
In the three different ways (see Appendix \ref{subsection:CAC_def}), $\woX$ can be uniquely represented by the initial values $\{x,\ox,\wx,X\}$ as
\begin{equation}
 \woX
 =\dfrac{Q_{m+1}(p_l-q_{m+1}\wx)\Big(({P_{l+1}}^2-{Q_m}^2)(P_l-Q_m X)\ox+Q_m(p_{l+1}-q_m \ox)(q_m-p_lx)X\Big)}
 {P_{l+1}(p_{l+1}-q_m\ox)\Big(({P_l}^2-{Q_{m+1}}^2)(P_l-Q_m X)x+P_l(q_m-p_l x)(p_l-q_{m+1}\wx)\Big)}.
\end{equation}
Since $\woX$ depends on $x$, the tetrahedron property does not hold.
Moreover, the condition \eqref{eqn:pq_PQ} can be obtained from the conditions \eqref{eqn:pq_albe_ga} and \eqref{eqn:PQ_albe_gala}.
Therefore, we have completed the proof.
\end{proof}

\begin{figure}[htbp]
\begin{center}
 \includegraphics[width=0.45\textwidth]{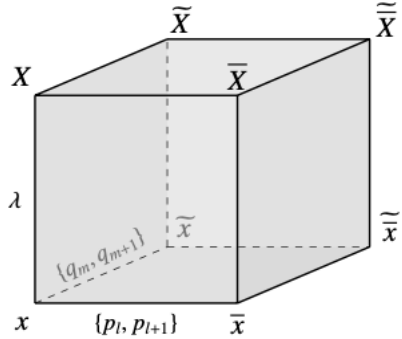}
\end{center}
\caption{
A cube for the CAC-system \eqref{eqns:lsG_CAC}.
See the caption of Figure \ref{fig:cube_ux} for explaining the placement of variables and parameters in the figure.
}
\label{fig:cube_xX}
\end{figure}

Setting
\begin{equation}
 X_{l,m}=\dfrac{F_{l,m}}{G_{l,m}},\quad
 \Psi_{l,m}=\begin{pmatrix}F_{l,m}\\G_{l,m}\end{pmatrix},
\end{equation}
from the equations \eqref{eqn:lsG_CAC_B} and \eqref{eqn:lsG_CAC_C}, we obtain the Lax pair \eqref{eqns:Lax_lsG} of the {\PDE} \eqref{eqn:lsG}.

\begin{remark}
Since Equation \eqref{eqn:lsG_CAC_AA} is also the lsG equation, setting
\begin{equation}
 x_{l,m}=\dfrac{F_{l,m}}{G_{l,m}},\quad
 \Psi_{l,m}=\begin{pmatrix}F_{l,m}\\G_{l,m}\end{pmatrix},
\end{equation}
we can construct another Lax pair of the lsG equation.
However, since the system \eqref{eqns:lsG_CAC} and the condition \eqref{eqn:pq_PQ} are invariant under the transformation
\begin{equation}
 x_{l,m}\leftrightarrow \dfrac{1}{X_{l,m}},\quad
 p_l\leftrightarrow P_l,\quad
 q_m\leftrightarrow Q_m,\quad
 \la\leftrightarrow-\la,
\end{equation}
the resulting Lax pair is essentially the same as the Lax pair \eqref{eqns:Lax_lsG}.
\end{remark}

\section{Integrable systems on the 3-dimensional integer lattice}\label{section:3D}
In general, if a quad-equation has a CAC-type auto-BT, it can be easily extended to a system of \PDE s on the 3-dimensional integer lattice $\bbZ^3$. 
(See \cite{ABS2003:MR1962121,HJN2016:MR3587455,NW2001:MR1869690}.)
However, in the case of non-auto-BTs, a little ingenuity is required.
(See \cite{BollR2011:MR2846098}.)
In this section, we show how quad-equations that have CABC-type BTs, which are non-auto-BTs, can be extended to a system of \PDE s on the lattice $\bbZ^3$.
Moreover, we also show how the resulting system can be reduced to systems on the lattice $\bbZ^3$ associated with the CAC-type BTs.
\subsection{A system of \PDE s on the lattice $\bbZ^3$ associated with the CABC property}
In this subsection, using the three CABC-systems \eqref{eqns:dKdV_CABC_1}, \eqref{eqns:lsG_CABC} and \eqref{eqns:dKdV_CABC_2},
we construct an integrable system of \PDE s on the lattice $\bbZ^3$.
We also show the $\tau$-function of the resulting system.

\begin{remark}
Since Equation \eqref{eqn:dKdV_CABC_1_S} can be obtained from the equations \eqref{eqn:dKdV_CABC_1_B} and \eqref{eqn:dKdV_CABC_1_C}, we here omit Equation \eqref{eqn:dKdV_CABC_1_S} when using the CABC-system \eqref{eqns:dKdV_CABC_1} for simplicity.
For the same reason, the equations \eqref{eqn:lsG_CABC_S} and \eqref{eqn:dKdV_CABC_2_S} are also omitted when using the CABC-systems \eqref{eqns:lsG_CABC} and \eqref{eqns:dKdV_CABC_2}, respectively.
\end{remark}

Let us denote the cubes for the three CABC-systems \eqref{eqns:dKdV_CABC_1}, \eqref{eqns:lsG_CABC} and \eqref{eqns:dKdV_CABC_2} by the symbols $C^{(1)}_{l,m}$, $C^{(2)}_{l,m}$ and $C^{(3)}_{l,m}$, respectively.
(See \S \ref{subsection:CABC_def} for the correspondence between a CABC-system and a cube.)
Overlap the cubes as in Figure \ref{fig_lattice_uvx} (Left), in the order $C^{(1)}_{l,m}$, $C^{(2)}_{l,m}$, $C^{(2)}_{l,m}$, $C^{(3)}_{l,m}$, so that the facing surfaces have the same equations.
We here also denote the resulting rectangular constructed by the set of four stacked cubes by the symbol $C^{(1223)}_{l,m}$.
Because both of the face-equations at the bottom and top of $C^{(1223)}_{l,m}$ are the dKdV equation, it can be placed repeatedly up and down as shown in Figure \ref{fig_lattice_uvx} (Right).
Note that since the parameter $\ga$ is not included in the dKdV equation, a different $\ga$ can be chosen for each stack of $C^{(1223)}_{l,m}$.
This repeated placement yields the following system of \PDE s on the lattice $\bbZ^3$:
\begin{subequations}\label{eqns:3D_uvx}
\begin{align}
 &\wou-u=\dfrac{\be_{m+1}-\al_l}{\wu}-\dfrac{\be_m-\al_{l+1}}{\ou},
 \label{eqn:3D_uvx_u}\\
 &\dfrac{\,\wox\,}{x}=\left(\dfrac{p_{l+1,n}-q_{m,n}\ox}{q_{m,n}-p_{l+1,n}\ox}\right)\left(\dfrac{q_{m+1,n}-p_{l,n}\wx}{p_{l,n}-q_{m+1,n}\wx}\right),
 \label{eqn:3D_uvx_x}\\
 &(u-v)\left(\dfrac{\be_m-\al_l}{u}+\wv\right)-\be_m+\ga_n=0,
 \label{eqn:3D_uvx_uvB}\\
 &\dfrac{\be_m-\al_l}{u}-\dfrac{\ga_n-\al_l}{v}-\ou+\ov=0,
 \label{eqn:3D_uvx_uvC}\\
 &\dfrac{v\wv}{{p_{l,n}}^2}-\dfrac{x(p_{l,n}-q_{m,n}x)}{p_{l,n}x-q_{m,n}}=0,
 \label{eqn:3D_uvx_xvB}\\
 &\dfrac{v\ov}{p_{l,n}p_{l+1,n}}-\dfrac{\ox(p_{l,n}-q_{m,n}x)}{p_{l,n}x-q_{m,n}}=0,
 \label{eqn:3D_uvx_xvC}\\
 &(\hu-\wv)\left(\dfrac{\be_m-\al_l}{\hu}+v\right)-\be_m+\ga_n=0,
 \label{eqn:3D_uvx_TuvB}\\
 &\dfrac{\be_m-\al_l}{\hu}-\dfrac{\ga_n-\al_{l+1}}{\ov}-\hou+v=0,
 \label{eqn:3D_uvx_TuvC}
\end{align}
\end{subequations}
where $u=u_{l,m,n}$, $v=v_{l,m,n}$, $x=x_{l,m,n}$ and 
\begin{equation}
 p_{l,n}=\sqrt{\al_l-\ga_n},\quad
 q_{m,n}=\sqrt{\be_m-\ga_n}.
\end{equation}

\begin{remark}
The correspondence between the cubes $C^{(1)}_{l,m}$, $C^{(2)}_{l,m}$, $C^{(3)}_{l,m}$ and the equations in the system \eqref{eqns:3D_uvx} is as follows.
{\rm
\begin{center}
\begin{tabular}{ll}
$C^{(1)}_{l,m}$ :&\eqref{eqn:3D_uvx_u}, \eqref{eqn:3D_uvx_uvB}, \eqref{eqn:3D_uvx_uvC}\\
$C^{(2)}_{l,m}$ :&\eqref{eqn:3D_uvx_x}, \eqref{eqn:3D_uvx_xvB}, \eqref{eqn:3D_uvx_xvC}\\
$C^{(3)}_{l,m}$ :&\eqref{eqn:3D_uvx_u}$_{n\to n+1}$, \eqref{eqn:3D_uvx_TuvB}, \eqref{eqn:3D_uvx_TuvC}
\end{tabular}
\end{center}
}
\end{remark}

\begin{figure}[htbp]
\begin{center}
 \includegraphics[width=0.4\textwidth]{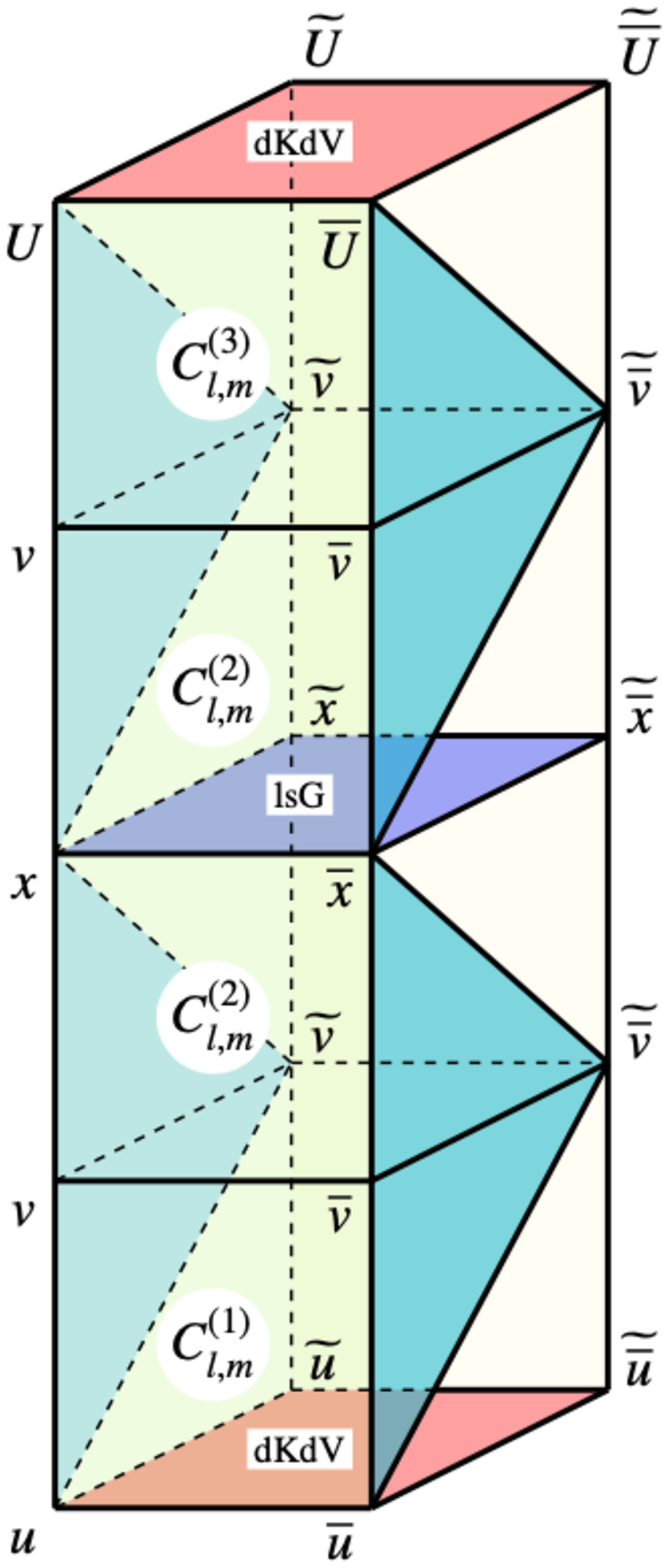}\hspace{3em}
 \includegraphics[width=0.42\textwidth]{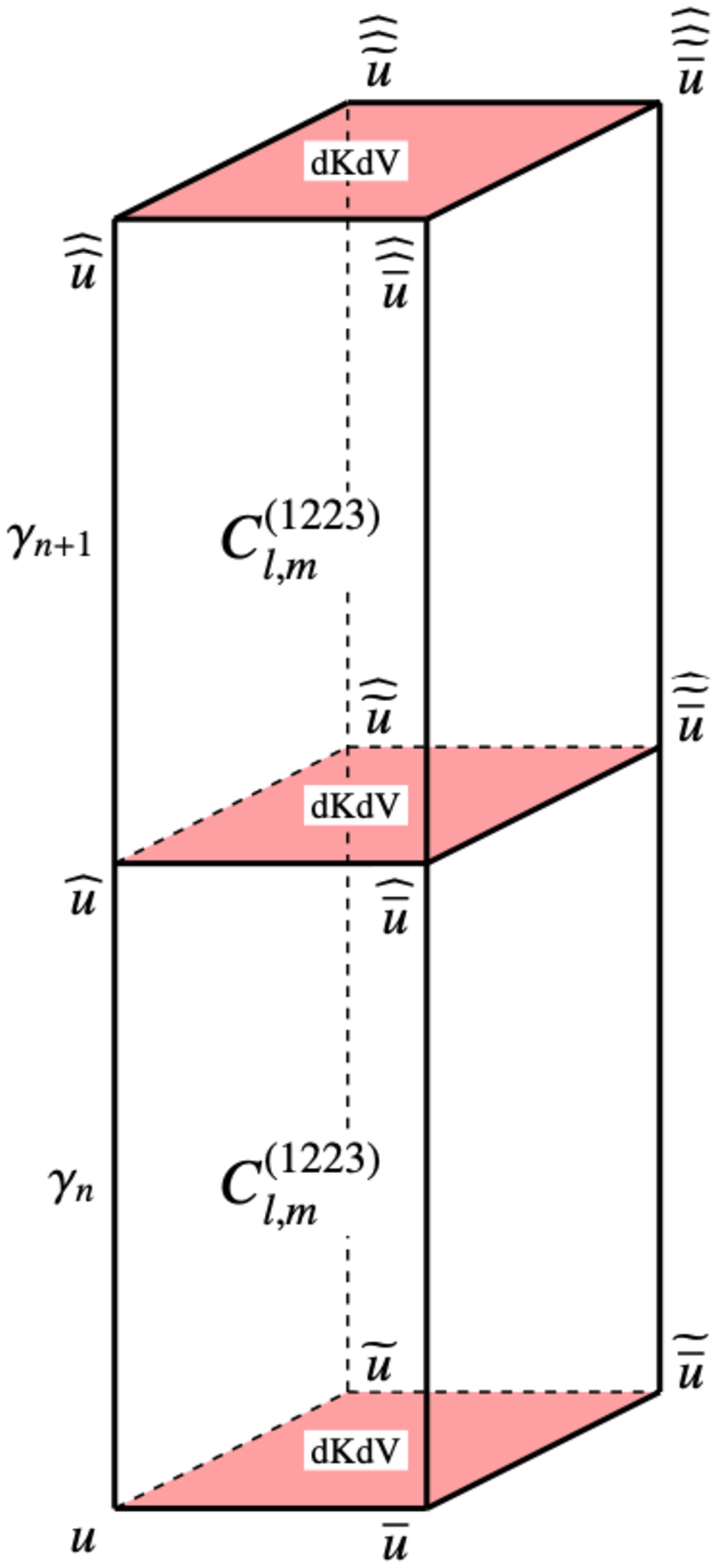}
\end{center}
\caption{
(Left): 
The rectangular $C^{(1223)}_{l,m}$ constructed by the four cubes.
Since the $u$-variable in the bottom cube $C^{(1)}_{l,m}$ and that in the top cube $C^{(3)}_{l,m}$ can be treated as different, we used the symbol $U$ for the $u$-variable of $C^{(3)}_{l,m}$ to differentiate them.\\
(Right): 
Two rectangulars $C^{(1223)}_{l,m}$.
Since we choose a different parameter $\ga$ for each $C^{(1223)}_{l,m}$, we place the parameters $\ga_n$ on the vertical line segments of the rectangles.
}
\label{fig_lattice_uvx}
\end{figure}

The $\tau$-function is one of the most important objects in the theory of integrable systems.
We show the $\tau$-function of the system \eqref{eqns:3D_uvx} in the following theorem.

\begin{theorem}\label{theorem:tau}
Let $\tau=\tau_{l,m,n}$ be the solution of the following bilinear equations:
\begin{subequations}\label{eqns:3D_uvx_bilinear}
\begin{align}
 &p_{l,n}\hwtau\,\otau-q_{m,n}\wtau\,\hotau+\tau\,\hwotau=0,\label{eqn:3D_uvx_bilinear_1}\\
 &q_{m,n}\hwtau\,\otau-p_{l,n}\wtau\,\hotau-\htau\,\wotau=0.\label{eqn:3D_uvx_bilinear_2}
\end{align}
\end{subequations}
Then, the $\tau$-function of the system \eqref{eqns:3D_uvx} is given as
\begin{equation}\label{eqn:3D_uvx_tau}
 u=-\dfrac{\,\tau\,\wotau\,}{\otau\,\wtau},\quad
 v=p_{l,n}\,\dfrac{\,\tau\,\hotau\,}{\otau\,\htau},\quad
 x=\dfrac{\,\hotau\,\wtau\,}{\otau\,\hwtau}.
\end{equation}
\end{theorem}
\begin{proof}
The proof is given in Appendix \ref{section:proof_tau}.
\end{proof}

\subsection{Systems of \PDE s on the lattice $\bbZ^3$ associated with the CAC property}
In this subsection, we derive integrable systems on the lattice $\bbZ^3$ associated with the CAC property from the system \eqref{eqns:3D_uvx}.

As in \S \ref{subsection:CABC_to_CAC_1} and \S \ref{subsection:CABC_to_CAC_2}, eliminating the $v$-variable from the system \eqref{eqns:3D_uvx} we obtain the following system of \PDE s depending only on the $u$- and $x$-variables:
\begin{subequations}\label{eqns:3D_ux}
\begin{align}
 &\wou-u=\dfrac{\be_{m+1}-\al_l}{\wu}-\dfrac{\be_m-\al_{l+1}}{\ou},
 \label{eqn:3D_ux_u}\\
 &\dfrac{\,\wox\,}{x}=\left(\dfrac{p_{l+1,n}-q_{m,n}\ox}{q_{m,n}-p_{l+1,n}\ox}\right)\left(\dfrac{q_{m+1,n}-p_{l,n}\wx}{p_{l,n}-q_{m+1,n}\wx}\right),
 \label{eqn:3D_ux_x}\\
 &\wu\,u+\dfrac{(p_{l,n}-q_{m,n}x)(q_{m+1,n}-p_{l,n}\wx)}{\wx}=0,
 \label{eqn:3D_ux_B1}\\
 &\ou\,u+\dfrac{(p_{l,n}-q_{m,n}x)(q_{m,n}-p_{l+1,n}\ox)}{x}=0,
 \label{eqn:3D_ux_C1}\\
 &\hwu~\hu+\dfrac{(q_{m,n}-p_{l,n}x)(p_{l,n}-q_{m+1,n}\wx)}{x}=0,
 \label{eqn:3D_ux_B2}\\
 &\hou~\hu+\dfrac{(q_{m,n}-p_{l,n}x)(p_{l+1,n}-q_{m,n}\ox)}{\ox}=0.
 \label{eqn:3D_ux_C2}
\end{align}
\end{subequations}
Indeed, eliminating $\wv$ from the equations \eqref{eqn:3D_uvx_uvB} and \eqref{eqn:3D_uvx_xvB},
we obtain
\begin{equation}
 v=\dfrac{p_{l,n}ux}{p_{l,n}x-q_{m,n}}.
\end{equation}
Then, eliminating $v$-variable from the equations \eqref{eqn:3D_uvx_uvB} and \eqref{eqn:3D_uvx_uvC} by using the equation above, we obtain the equations \eqref{eqn:3D_ux_B1} and \eqref{eqn:3D_ux_C1}, respectively.
Moreover, eliminating $\wv$ from the equations \eqref{eqn:3D_uvx_xvB} and \eqref{eqn:3D_uvx_TuvB},
we obtain
\begin{equation}
 v=\dfrac{p_{l,n}(p_{l,n}-q_{m,n}x)}{\hu},
\end{equation}
and then eliminating $v$-variable from the equations \eqref{eqn:3D_uvx_TuvB} and \eqref{eqn:3D_uvx_TuvC} by using the equation above, we obtain the equations \eqref{eqn:3D_ux_B2} and \eqref{eqn:3D_ux_C2}, respectively.

\begin{remark}
As shown in Figure \ref{fig:lattice_ux_x} (Left), the system \eqref{eqns:3D_ux} can be thought of as a superposition of the CAC-systems \eqref{eqns:dKdV_lsG_CAC_1} and \eqref{eqns:dKdV_lsG_CAC_2} alternating in the $n$-direction.
\end{remark}

\begin{figure}[htbp]
\begin{center}
 \includegraphics[width=0.45\textwidth]{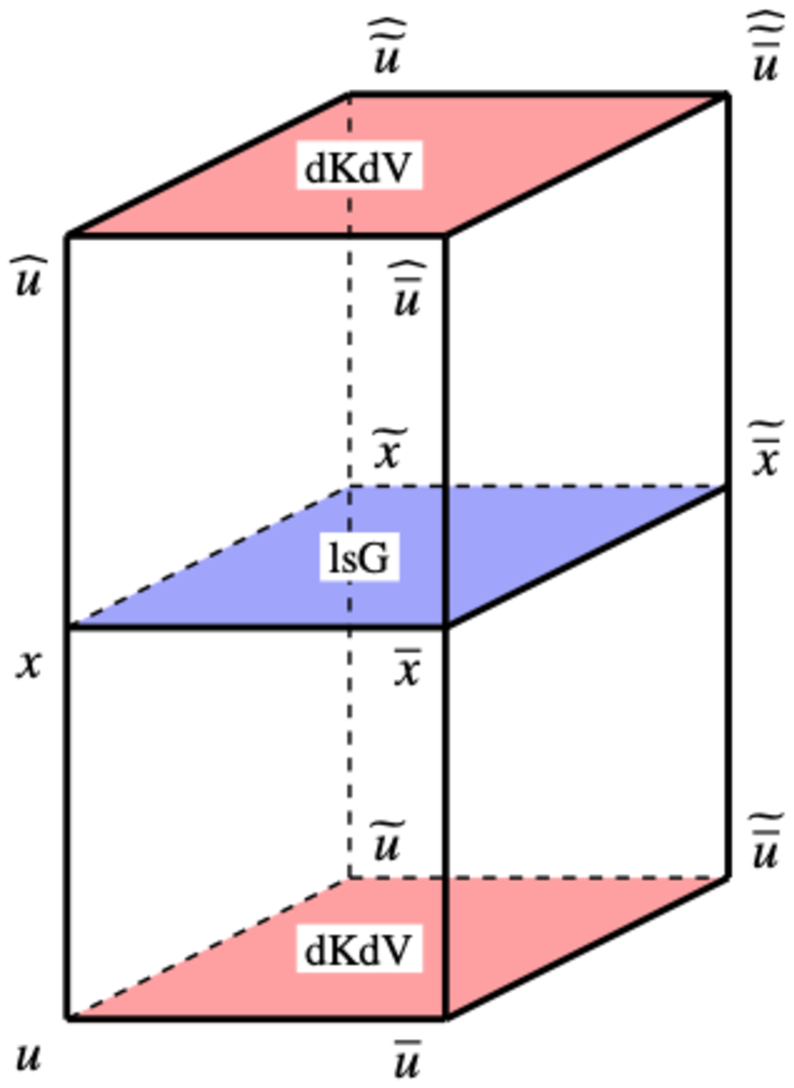}\hspace{2em}
 \includegraphics[width=0.45\textwidth]{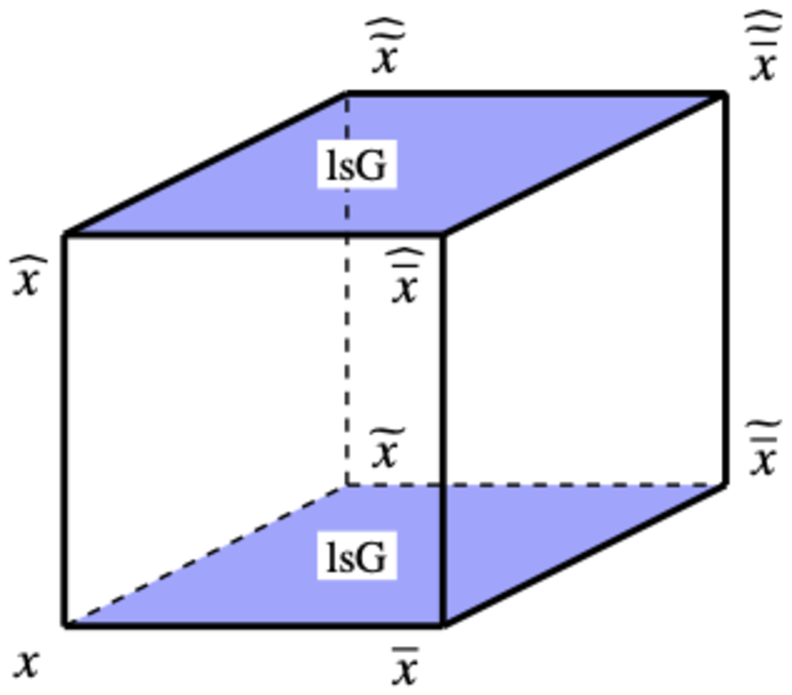}
\end{center}
\caption{
(Left): 
Cubes for the CAC-system \eqref{eqns:3D_ux}.
The system of equations \eqref{eqn:3D_ux_u}, \eqref{eqn:3D_ux_x}, \eqref{eqn:3D_ux_B1} and \eqref{eqn:3D_ux_C1} correspond to the lower cube, 
and the system of equations \eqref{eqn:3D_ux_u}$_{n\to n+1}$, \eqref{eqn:3D_ux_x}, \eqref{eqn:3D_ux_B2} and \eqref{eqn:3D_ux_C2} correspond to the upper cube.\\
(Right): 
A cube for the CAC-system \eqref{eqns:3D_x}.
}
\label{fig:lattice_ux_x}
\end{figure}

As in \S \ref{subsection:2CACs_to_CAC}, 
we eliminate $\hwu\,\hu$ from the equations \eqref{eqn:3D_ux_B1}$_{n\to n+1}$ and \eqref{eqn:3D_ux_B2}.
Also, we eliminate $\hou\,\hu$ from the equations \eqref{eqn:3D_ux_C1}$_{n\to n+1}$ and \eqref{eqn:3D_ux_C2}.
The resulting equations, together with Equation \eqref{eqn:3D_ux_x}, lead to the following system of \PDE s depending only on the $x$-variable:
\begin{subequations}\label{eqns:3D_x}
\begin{align}
 &\dfrac{\,\wox\,}{x}=\left(\dfrac{p_{l+1,n}-q_{m,n}\ox}{q_{m,n}-p_{l+1,n}\ox}\right)\left(\dfrac{q_{m+1,n}-p_{l,n}\wx}{p_{l,n}-q_{m+1,n}\wx}\right),
 \label{eqn:3D_x_x}\\
 &\dfrac{(q_{m,n}-p_{l,n}x)(p_{l,n}-q_{m+1,n}\wx)}{x}
 =\dfrac{(p_{l,n+1}-q_{m,n+1}\hx)(q_{m+1,n+1}-p_{l,n+1}\hwx)}{\hwx},
 \label{eqn:3D_x_B}\\
 &\dfrac{(q_{m,n}-p_{l,n}x)(p_{l+1,n}-q_{m,n}\ox)}{\ox}
 =\dfrac{(p_{l,n+1}-q_{m,n+1}\hx)(q_{m,n+1}-p_{l+1,n+1}\hox)}{\hx}.
 \label{eqn:3D_x_C}
\end{align}
\end{subequations}

\vspace{0.5em}
\begin{remark}
~\\[-1em]
\begin{enumerate}
\item 
As shown in Figure \ref{fig:lattice_ux_x} (Right), the system \eqref{eqns:3D_x} can be thought of as a superposition of the CAC-system \eqref{eqns:lsG_CAC} in the $n$-direction.
\item
Both equations \eqref{eqn:3D_x_B} and \eqref{eqn:3D_x_C} are also the lsG equation.
Indeed, by setting
\begin{equation}
 y_{m,n}=\dfrac{q_{m,n}-p_{l,n}x_{l,m,n}}{\sqrt{\al_l-\be_m}\,x_{l,m,n}},\quad
 a_m=-\sqrt{\al_l-\be_m},\quad
 b_n=p_{l,n},
\end{equation}
Equation \eqref{eqn:3D_x_B} can be rewritten as the following standard form of the lsG equation in $(m,n)$-direction:
\begin{equation}
 \dfrac{y_{m+1,n+1}}{y_{m,n}}
 =\left(\dfrac{b_{n+1}-a_my_{m,n+1}}{a_m-b_{n+1}y_{m,n+1}}\right)
 \left(\dfrac{a_{m+1}-b_ny_{m+1,n}}{b_n-a_{m+1}y_{m+1,n}}\right),
\end{equation}
while by setting
\begin{equation}
 z_{l,n}=\dfrac{p_{l,n}x_{l,m,n}-q_{m,n}}{\sqrt{\be_m-\al_l}},\quad
 c_l=-\sqrt{\be_m-\al_l},\quad
 d_n=q_{m,n},
\end{equation}
Equation \eqref{eqn:3D_x_C} can also be rewritten as the following standard form of the lsG equation in $(l,n)$-direction:
\begin{equation}
 \dfrac{z_{l+1,n+1}}{z_{l,n}}
 =\left(\dfrac{d_{n+1}-c_lz_{l,n+1}}{c_l-d_{n+1}z_{l,n+1}}\right)
 \left(\dfrac{c_{l+1}-d_nz_{l+1,n}}{d_n-c_{l+1}z_{l+1,n}}\right).
\end{equation}
\end{enumerate}
\end{remark}

\section{Concluding remarks}\label{ConcludingRemarks}
In this paper, using the CABC properties of the dKdV equation \eqref{eqn:dKdV} and the lsG equation \eqref{eqn:lsG}, we have shown that they have the CAC property.
Using the CABC and CAC properties, we also obtained their new Lax pairs.
Moreover, we have constructed the integrable system \eqref{eqns:3D_uvx} on the lattice $\bbZ^3$ by stacking the cubes associated with the CABC property in the third direction and have clarified the structure of its $\tau$-function.

In addition to the dKdV and lsG equations, there exist other integrable \PDE s for which we still do not know whether they have the CAC property.
We have shown that it is useful to consider the CABC property to clarify the CAC property of \PDE s.
Many studies have been done to find \PDE s which have the CAC property (see \S \ref{Introduction}). 
However, the dKdV and lsG equations have the CAC property but are not included in these previous studies.
For the reasons above, classifying \PDE s with the CABC property may be important for clarifying the CAC property of known integrable \PDE s. 
Moreover, it may lead to finding new integrable \PDE s. 
This should be a subject of future research.
\subsection*{Acknowledgment}
The author would like to thank Dr. Pavlos Kassotakis and Dr. Yang Shi for fruitful discussions.
This research was supported by a JSPS KAKENHI Grant Number JP19K14559.
The author used Wolfram Mathematica for checking computations.
\appendix
\section{The CAC and CABC properties}\label{section:CAC_CABC}
In this appendix, we explain the consistency around a cube (CAC) property of a system of \PDE s in a particular form.
In addition, we also explain the consistency around a broken cube (CABC) property of a system of \PDE s in a particular form.
For more general definitions, see \cite{BS2002:MR1890049,HJN2016:MR3587455,NijhoffFW2002:MR1912127,WalkerAJ:thesis} for the CAC property and see \cite{JN2021:zbMATH07476241,nakazono2022properties} for the CABC property.

\begin{remark}
For simplicity, we omit the parameters of the equations in this appendix.
Moreover, the equations in this appendix are not limited to autonomous type.
For example, the relation between $Q(u,\ou,\wu,\wou)=0$ and $Q(\ou,\overline{\ou},\wou,\widetilde{\overline{\ou}})=0$ is not simple replacement of the corresponding $u$-variables; applying $l\to l+1$ to $Q(u,\ou,\wu,\wou)=0$ gives $Q(\ou,\overline{\ou},\wou,\widetilde{\overline{\ou}})=0$.
See \S \ref{section:dKdV_lsG_consistency} for examples involving parameters.
\end{remark}

\subsection{The CAC property}\label{subsection:CAC_def}
Let us consider the following system of \PDE s:
\begin{subequations}\label{eqns:CAC_PDE_ABCA'}
\begin{align}
 &A(u,\ou,\wu,\wou)=0,\label{eqn:CAC_PDE_A}\\
 &B(u,\wu,v,\wv)=0,\label{eqn:CAC_PDE_B}\\
 &C(u,\ou,v,\ov)=0,\label{eqn:CAC_PDE_C}\\
 &A'(v,\ov,\wv,\wov)=0,\label{eqn:CAC_PDE_AA}
\end{align}
\end{subequations}
where $u=u_{l,m}$ and $v=v_{l,m}$.
Here, the functions $A$, $B$, $C$ and $A'$ are irreducible multilinear polynomials in four variables.
Consider the following sublattice of the lattice $\bbZ^3$:
\begin{equation}\label{eqn:lattice_01}
 \set{(l,m,0)\in\bbZ^3}{l,m\in\bbZ}\cup\set{(l,m,1)\in\bbZ^3}{l,m\in\bbZ}
\end{equation}
and assign the $u$- and $v$-variables on the vertices of the sublattice by the following correspondences:
\begin{equation}\label{eqn:lattice_uv}
 (l,m,0)\,\leftrightarrow\,u_{l,m},\qquad
 (l,m,1)\,\leftrightarrow\,v_{l,m}.
\end{equation}
Then, focusing on the cube given by the 8 points
\begin{align}
 &(l,m,0),\quad (l+1,m,0),\quad (l,m+1,0),\quad (l+1,m+1,0),\notag\\
 &(l,m,1),\quad (l+1,m,1),\quad (l,m+1,1),\quad (l+1,m+1,1),
 \label{eqn:lattice_8points}
\end{align}
from the system \eqref{eqns:CAC_PDE_ABCA'}, we obtain the following face-equations of the cube (see Figure \ref{fig:CACcube}):
\begin{subequations}\label{eqns:CAC_ABCA'}
\begin{align}
 &{\mathcal A}=A(u,\ou,\wu,\wou)=0,\label{eqn:CAC_A}\\
 &{\mathcal A}'=A'(v,\ov,\wv,\wov)=0,\label{eqn:CAC_AA}\\
 &{\mathcal B}=B(u,\wu,v,\wv)=0,\label{eqn:CAC_B}\\
 &\overline{{\mathcal B}}=B(\ou,\wou,\ov,\wov)=0,\label{eqn:CAC_BB}\\
 &{\mathcal C}=C(u,\ou,v,\ov)=0,\label{eqn:CAC_C}\\
 &\widetilde{{\mathcal C}}=C(\wu,\wou,\wv,\wov)=0.\label{eqn:CAC_CC}
\end{align}
\end{subequations}
The CAC and tetrahedron properties are defined as follows.
\begin{enumerate}
\item 
There are following three ways to calculate $\wov$ by using all the equations in the system \eqref{eqns:CAC_ABCA'} with $\{u,\ou,\wu,v\}$ as initial values.
\begin{enumerate}
\item 
Express $\wov$ as a rational function in terms of $\{v,\ov,\wv\}$ by using Equation \eqref{eqn:CAC_AA}.
Then, eliminate $\wv$ by using Equation \eqref{eqn:CAC_B} and $\ov$ by using Equation \eqref{eqn:CAC_C} from it.\\[-0.7em]
\item 
Express $\wov$ as a rational function in terms of $\{\ou,\wou,\ov\}$ by using Equation \eqref{eqn:CAC_BB}.
Then, eliminate $\wou$ by using Equation \eqref{eqn:CAC_A} and $\ov$ by using Equation \eqref{eqn:CAC_C} from it.\\[-0.7em]
\item 
Express $\wov$ as a rational function in terms of $\{\wu,\wou,\wv\}$ by using Equation \eqref{eqn:CAC_CC}.
Then, eliminate $\wou$ by using Equation \eqref{eqn:CAC_A} and $\wv$ by using Equation \eqref{eqn:CAC_B} from it.\\[-0.7em]
\end{enumerate}
When $\wov$ is uniquely determined as a rational function with the initial values $\{u,\ou,\wu,v\}$, then the system \eqref{eqns:CAC_PDE_ABCA'} is said to have the {\it CAC property} or said to be a {\it CAC-system}.
\item
Let the system \eqref{eqns:CAC_PDE_ABCA'} be a CAC-system.
When $\wov$ is represented as a rational function with the initial values $\{u,\ou,\wu,v\}$, the system \eqref{eqns:CAC_PDE_ABCA'} is said to have the {\it tetrahedron property} if the rational function does not depend on $u$.
\item
We say that the {\PDE} \eqref{eqn:CAC_PDE_A} (or the {\PDE} \eqref{eqn:CAC_PDE_AA}) has the CAC property if the system \eqref{eqns:CAC_PDE_ABCA'} has the CAC property. 
The pair of equations \eqref{eqn:CAC_PDE_B} and \eqref{eqn:CAC_PDE_C} is then referred to as a {\it CAC-type BT} from the {\PDE} \eqref{eqn:CAC_PDE_A} to the {\PDE} \eqref{eqn:CAC_PDE_AA} (or from the {\PDE} \eqref{eqn:CAC_PDE_AA} to the {\PDE} \eqref{eqn:CAC_PDE_A}).
When the equations \eqref{eqn:CAC_PDE_A} and \eqref{eqn:CAC_PDE_AA} are the same {\PDE}, the CAC-type BT \eqref{eqn:CAC_PDE_B} and \eqref{eqn:CAC_PDE_C} is specifically called a {\it CAC-type auto-BT} of the {\PDE} \eqref{eqn:CAC_PDE_A}.
\end{enumerate}

\begin{remark}
Even if Equation \eqref{eqn:CAC_PDE_A} has the CAC property, it is not necessarily integrable\cite{HietarintaJ2019:zbMATH07053246}.
Even if a Lax pair of Equation \eqref{eqn:CAC_PDE_A} is constructed from the CAC-type BT \eqref{eqn:CAC_PDE_B} and \eqref{eqn:CAC_PDE_C} by the method in \cite{BS2002:MR1890049,HJN2016:MR3587455,NijhoffFW2002:MR1912127,WalkerAJ:thesis}, it may be fake \cite{butler2015two} or weak \cite{HV2012:zbMATH06063057}.
From this perspective, in \cite{HietarintaJ2019:zbMATH07053246}, the CAC-type BT \eqref{eqn:CAC_PDE_B} and \eqref{eqn:CAC_PDE_C} are characterized as follows.
Calculate $\wov$ in the following two ways.
\begin{enumerate}
\item[{\rm (a)}]
Express $\wov$ as a rational function in terms of $\{\ou,\wou,\ov\}$ by using Equation \eqref{eqn:CAC_BB}.
Then, eliminate $\ov$ from it by using Equation \eqref{eqn:CAC_C}.
Finally, let $\wov=f_1(u,\ou,\wou,v)$ be the resulting relation.
\item[{\rm (b)}]
Express $\wov$ as a rational function in terms of $\{\wu,\wou,\wv\}$ by using Equation \eqref{eqn:CAC_CC}.
Then, eliminate $\wv$ from it by using Equation \eqref{eqn:CAC_B}.
Finally, let $\wov=f_2(u,\wu,\wou,v)$ be the resulting relation.
\end{enumerate}
If $f_1(u,\ou,\wou,v)=f_2(u,\wu,\wou,v)$ is automatically equal, then the CAC-type BT \eqref{eqn:CAC_PDE_B} and \eqref{eqn:CAC_PDE_C} is said to be {\it trivial}.
If only Equation \eqref{eqn:CAC_A} is obtained from $f_1(u,\ou,\wou,v)=f_2(u,\wu,\wou,v)$, then the CAC-type BT \eqref{eqn:CAC_PDE_B} and \eqref{eqn:CAC_PDE_C} is said to be {\it strong}.
If $f_1(u,\ou,\wou,v)=f_2(u,\wu,\wou,v)$ yields Equation \eqref{eqn:CAC_A} plus other \PDE s for $\{u,\ou,\wu,\wou\}$, then the CAC-type BT \eqref{eqn:CAC_PDE_B} and \eqref{eqn:CAC_PDE_C} is said to be {\it weak}.
A trivial BT gives a fake Lax pair, and a weak BT gives a weak Lax pair.
It is claimed in \cite{HietarintaJ2019:zbMATH07053246} that if the CAC-type BT is strong, then the integrability of \eqref{eqn:CAC_PDE_A} should be guaranteed.

Note that the CAC-type BTs given in Theorems \ref{thm:CAC1}, \ref{thm:CAC2} and \ref{thm:CAC3} are all strong.
\end{remark}

\begin{figure}[htbp]
\begin{center}
 \includegraphics[width=0.5\textwidth]{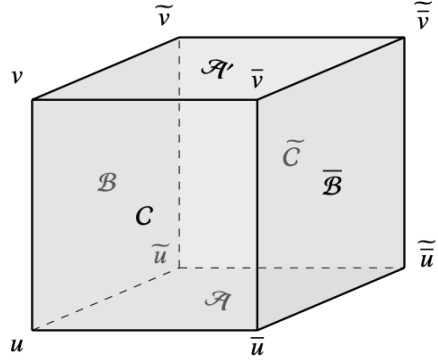}
\end{center}
\caption{A cube for the system \eqref{eqns:CAC_PDE_ABCA'}.
The $u$- and $v$-variables are assigned at the bottom and top vertices of the cube, respectively.
Also, the equations \eqref{eqns:CAC_ABCA'} are assigned to the faces of the cube.}
\label{fig:CACcube}
\end{figure}

\subsection{The CABC property}\label{subsection:CABC_def}
Let us consider the following system of \PDE s:
\begin{subequations}\label{eqns:CABC_PDE_ASBC}
\begin{align}
 &A(u,\ou,\wu,\wou)=0,\label{eqn:CABC_PDE_A}\\
 &S(u,\ou,\wv,\wov)=0,\label{eqn:CABC_PDE_S}\\
 &B(u,v,\wv)=0,\label{eqn:CABC_PDE_B}\\
 &C(u,\ou,v,\ov)=0,\label{eqn:CABC_PDE_C}
\end{align}
\end{subequations}
where $u=u_{l,m}$ and $v=v_{l,m}$.
Here, the functions $A$, $S$ and $C$ are irreducible multilinear polynomials in four variables.
The function $B=B(x,y,z)$ is a polynomial in three variables satisfying the following:
\begin{enumerate}
\item[1)]
$\deg_x B\geq 1$,\quad
$\deg_y B=\deg_z B=1$;
\item[2)]
Let $y=f(x,z)$ be the solution of $B=0$.
Then, $f(x,z)$ is a rational function that depends on $x$ and $z$, that is, 
the following hold:
\begin{equation}
 \dfrac{\partial}{\partial x} f(x,z)\neq0,\quad
 \dfrac{\partial}{\partial z} f(x,z)\neq0.
\end{equation}
\end{enumerate}
Assign the $u$- and $v$-variables on the sublattice \eqref{eqn:lattice_01} by the correspondence \eqref{eqn:lattice_uv}.
Then, by considering the cube consisting of the eight vertices \eqref{eqn:lattice_8points} (see Figure \ref{fig:CABCcube}), the system of equations around the cube obtained from the system \eqref{eqns:CABC_PDE_ASBC} is the following:
\begin{subequations}\label{eqns:CABC_ASBC}
\begin{align}
 &{\mathcal A}=A(u,\ou,\wu,\wou)=0,\label{eqn:CABC_A}\\
 &{\mathcal S}=S(u,\ou,\wv,\wov)=0,\label{eqn:CABC_S}\\
 &{\mathcal B}=B(u,v,\wv)=0,\label{eqn:CABC_B}\\
 &\overline{{\mathcal B}}=B(\ou,\ov,\wov)=0,\label{eqn:CABC_BB}\\
 &{\mathcal C}=C(u,\ou,v,\ov)=0,\label{eqn:CABC_C}\\
 &\widetilde{{\mathcal C}}=C(\wu,\wou,\wv,\wov)=0.\label{eqn:CABC_CC}
\end{align}
\end{subequations}

\begin{figure}[htbp]
\begin{center}
 \includegraphics[width=0.5\textwidth]{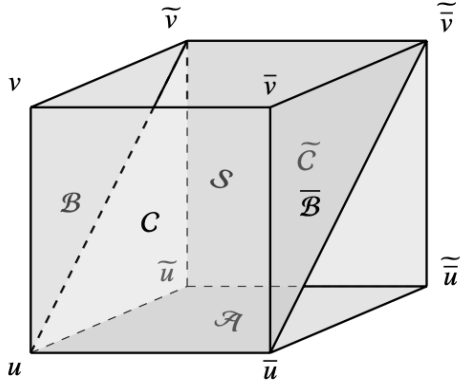}
\end{center}
\caption{A cube for the system \eqref{eqns:CABC_PDE_ASBC}.
The $u$- and $v$-variables are assigned at the bottom and top vertices of the cube, respectively.
Also, each of the equations \eqref{eqns:CABC_ASBC} is assigned to a face of the cube.
Note that the face for Equation \eqref{eqn:CABC_S} corresponds to the cutting plane bisecting the cube diagonally, and the faces for the equations \eqref{eqn:CABC_B} and \eqref{eqn:CABC_BB} correspond to the upper triangles.}
\label{fig:CABCcube}
\end{figure}

The CABC and tetrahedron properties are defined as follows.
\begin{enumerate}
\item 
There are following three ways to calculate $\wov$ by using all the equations in the system \eqref{eqns:CABC_ASBC} with $\{u,\ou,\wu,v\}$ as initial values.
\begin{enumerate}
\item 
Express $\wov$ as a rational function in terms of $\{u,\ou,\wv\}$ by using Equation \eqref{eqn:CABC_S}.
Then, eliminate $\wv$ from it by using Equation \eqref{eqn:CABC_B}.\\[-0.7em]
\item 
Express $\wov$ as a rational function in terms of $\{\ou,\ov\}$ by using Equation \eqref{eqn:CABC_BB}.
Then, eliminate $\ov$ from it by using Equation \eqref{eqn:CABC_C}.\\[-0.7em]
\item 
Express $\wov$ as a rational function in terms of $\{\wu,\wou,\wv\}$ by using Equation \eqref{eqn:CABC_CC}.
Then, eliminate $\wou$ by using Equation \eqref{eqn:CABC_A} and $\wv$ by using Equation \eqref{eqn:CABC_B} from it.\\[-0.7em]
\end{enumerate}
When $\wov$ is uniquely determined as a rational function with the initial values $\{u,\ou,\wu,v\}$, 
then the system \eqref{eqns:CABC_PDE_ASBC} is said to have the {\it CABC property} or said to be a {\it CABC-system};
the {\PDE} \eqref{eqn:CABC_PDE_A} is said to have the CABC property. 
\item 
Let the system \eqref{eqns:CABC_PDE_ASBC} be a CABC-system.
Then, the system \eqref{eqns:CABC_PDE_ASBC} is said to have the {\it tetrahedron property},
if the equations
\begin{equation}
 K_1(u,\ou,v,\wov)=0,\quad
 K_2(u,\ou,\ov,\wv)=0,
\end{equation}
where $K_1$ and $K_2$ are irreducible multilinear polynomials in four variables, can be obtained from the system \eqref{eqns:CABC_PDE_ASBC}.
\item 
If there exists a {\PDE} given only by $\{v,\ov,\wv,\wov\}$, then the tuple of equations \eqref{eqn:CABC_PDE_S}, \eqref{eqn:CABC_PDE_B} and \eqref{eqn:CABC_PDE_C} is referred to as a {\it CABC-type BT} from the {\PDE} \eqref{eqn:CABC_PDE_A} to the {\PDE} given by $\{v,\ov,\wv,\wov\}$.
\end{enumerate}

\section{Lax pairs of the dKdV equation and the lsG equation}\label{section:Lax}
In this appendix, we show Lax pairs of the dKdV equation \eqref{eqn:dKdV} and the lsG equation \eqref{eqn:lsG} obtained in this study.

\subsection{Lax pairs of the {\PDE} \eqref{eqn:dKdV}}
A Lax pair of the {\PDE} \eqref{eqn:dKdV} is given by
\begin{equation}\label{eqn:Lax_phi}
 \Phi_{l+1,m}=L_{l,m}\Phi_{l,m},\quad
 \Phi_{l,m+1}=M_{l,m}\Phi_{l,m},
\end{equation}
where the matrices $L_{l,m}$ and $M_{l,m}$ are given by one of the following pairs:
\begin{subequations}\label{eqns:Lax_dKdV}
\begin{align}
 &L_{l,m}=
 \begin{pmatrix}
  0&\ga-\al_{l+1}\\1&\dfrac{\be_m-\al_l}{u}-\ou
 \end{pmatrix},\quad
  M_{l,m}=
 \begin{pmatrix}
 u&\ga-\al_l\\1&\dfrac{\be_m-\al_l}{u}
 \end{pmatrix},
 \label{eqn:Lax_dKdV_CABC}\\
 & L_{l,m}=
 \begin{pmatrix}
  \dfrac{{q_m}^2}{u}-\ou&-\dfrac{p_lq_m}{u}\\[0.5em]
  \dfrac{p_{l+1}q_m}{u}&-\dfrac{p_{l+1}p_l}{u}
 \end{pmatrix},\quad
 M_{l,m}=
 \begin{pmatrix}
  \dfrac{q_{m+1}q_m}{u}&-\dfrac{p_lq_{m+1}}{u}\\[0.5em]
  \dfrac{p_lq_m}{u}&\wu-\dfrac{{p_l}^2}{u}
 \end{pmatrix},
 \label{eqn:Lax_dKdV_CAC1}\\
  &L_{l,m}=
 \begin{pmatrix}
  \dfrac{p_{l+1}p_l}{u}&-\dfrac{p_{l+1}q_m}{u}\\[0.5em]
  \dfrac{p_lq_m}{u}&\ou-\dfrac{{q_m}^2}{u}
 \end{pmatrix},\quad
 M_{l,m}=
 \begin{pmatrix}
  \dfrac{{p_l}^2}{u}-\wu&-\dfrac{p_lq_m}{u}\\[0.5em]
  \dfrac{p_lq_{m+1}}{u}&-\dfrac{q_{m+1}q_m}{u}
 \end{pmatrix}.
 \label{eqn:Lax_dKdV_CAC2}
\end{align}
\end{subequations}
Here, $p_l$ and $q_m$ are given by Equation \eqref{eqn:pq_albe_ga} and $\ga\in\bbC$ is a spectral parameter.
Indeed, for each pair in \eqref{eqns:Lax_dKdV}, the compatibility condition
\begin{equation}
 L_{l,m+1}M_{l,m}=M_{l+1,m}L_{l,m}
\end{equation}
gives the {\PDE} \eqref{eqn:dKdV}.
\subsection{A Lax pair of the {\PDE} \eqref{eqn:lsG}}
The linear system 
\begin{subequations}\label{eqns:Lax_lsG}
\begin{equation}
 \Psi_{l+1,m}={\mathcal L}_{l,m}\Psi_{l,m},\quad
 \Psi_{l,m+1}={\mathcal M}_{l,m}\Psi_{l,m}
\end{equation}
where
\begin{align}
 &{\mathcal L}_{l,m}=
 \begin{pmatrix}
 p_{l+1}-q_m\ox+\dfrac{{Q_m}^2\ox}{q_m-p_lx}&\dfrac{P_lQ_m\ox}{p_lx-q_m}\\[1em]
 \dfrac{P_{l+1}Q_m\ox}{q_m-p_lx}&\dfrac{P_lP_{l+1}\ox}{p_lx-q_m}
 \end{pmatrix},\\[1em]
 &{\mathcal M}_{l,m}=
 \begin{pmatrix}
  \dfrac{Q_mQ_{m+1}}{q_m-p_lx}&\dfrac{P_lQ_{m+1}}{p_lx-q_m}\\[1em]
  \dfrac{P_lQ_m}{q_m-p_lx}&\dfrac{{P_l}^2}{p_lx-q_m}+\dfrac{q_{m+1}\wx-p_l}{x}
 \end{pmatrix},
\end{align}
\end{subequations}
is a Lax pair of the {\PDE} \eqref{eqn:lsG}.
Here, $P_l$ and $Q_m$ are given by Equation \eqref{eqn:pq_PQ} and $\la\in\bbC$ is a spectral parameter.
Indeed, the compatibility condition
\begin{equation}
 {\mathcal L}_{l,m+1}{\mathcal M}_{l,m}={\mathcal M}_{l+1,m}{\mathcal L}_{l,m}
\end{equation}
gives the {\PDE} \eqref{eqn:lsG}.
\section{Proof of Theorem \ref{theorem:tau}}\label{section:proof_tau}
Define $\omega=\omega_{l,m,n}$ by
\begin{equation}
 \omega=\dfrac{\,\htau\,}{\tau},
\end{equation}
where $\tau=\tau_{l,m,n}$ is a function satisfying the bilinear equations \eqref{eqns:3D_uvx_bilinear}.
Then, using Equation \eqref{eqn:3D_uvx_tau}, we can express the variables $u=u_{l,m,n}$, $v=v_{l,m,n}$ and $x=x_{l,m,n}$ in the $\omega$-variable as
\begin{subequations}\label{eqns:uvx_omega}
\begin{align}
 &u=-\dfrac{\,\tau\,\wotau\,}{\otau\,\wtau}
 =\cfrac{p_{l,n}\dfrac{\,\hotau\,}{\otau}-q_{m,n}\dfrac{\,\hwtau\,}{\wtau}}{\dfrac{\,\htau\,}{\tau}}
 =\dfrac{p_{l,n}\oomega-q_{m,n}\womega}{\omega}
 \label{eqn:uvx_omega_u},\\
 &v=p_{l,n}\,\dfrac{\,\tau\,\hotau\,}{\otau\,\htau}=p_{l,n}\dfrac{\,\oomega\,}{\omega},\\
 &x=\dfrac{\,\hotau\,\wtau\,}{\otau\,\hwtau}=\dfrac{\,\oomega\,}{\womega}.
\end{align}
\end{subequations}
Note that we use Equation \eqref{eqn:3D_uvx_bilinear_2} for the deformation in Equation \eqref{eqn:uvx_omega_u}.
Moreover, the following lemma holds.
\begin{lemma}
The following relations hold:
\begin{subequations}\label{eqns:omega}
\begin{align}
 & \dfrac{\,\woomega\,}{\omega}=\dfrac{q_{m,n}\oomega-p_{l,n}\womega}{q_{m,n}\womega-p_{l,n}\oomega},
 \label{eqn:omega_1}\\
 &\dfrac{\omega\,\homega}{\oomega\,\womega}
 =\dfrac{p_{l,n+1}\hoomega-q_{m,n+1}\hwomega}{p_{l,n}\womega-q_{m,n}\oomega}.
 \label{eqn:omega_2}
\end{align}
\end{subequations}
\end{lemma}
\begin{proof}
Using Equation \eqref{eqn:3D_uvx_bilinear_1}, we obtain
\begin{equation}
 \woomega
 =\dfrac{\,\hwotau\,}{\wotau}
 =\dfrac{q_{m,n}\hotau~\wtau-p_{l,m}\otau~\hwtau}{\wotau~\tau},
\end{equation}
and using Equation \eqref{eqn:3D_uvx_bilinear_2}, we obtain
\begin{equation}
 \dfrac{\omega(q_{m,n}\oomega-p_{l,m}\womega)}{q_{m,n}\womega-p_{l,m}\oomega}
 =\dfrac{\htau(q_{m,n}\hotau~\wtau-p_{l,m}\otau~\hwtau)}{\tau(q_{m,n}\otau~\hwtau-p_{l,m}\wtau~\hotau)}
 =\dfrac{q_{m,n}\hotau~\wtau-p_{l,m}\otau~\hwtau}{\wotau~\tau}.
\end{equation}
Equation \eqref{eqn:omega_1} follows from these two equations.
Moreover, using the equations \eqref{eqn:3D_uvx_bilinear_2}$_{n\to n+1}$ and \eqref{eqn:3D_uvx_bilinear_1} in turn, we obtain
\begin{align}
 \dfrac{p_{l,n+1}\hoomega-q_{m,n+1}\hwomega}{\homega(p_{l,n}\womega-q_{m,n}\oomega)}
 &=\dfrac{\otau~\wtau~\htau(p_{l,n+1}\hwtau~\widehat{\hotau}-q_{m,n+1}\hotau~\widehat{\hwtau})}{\hotau~\hwtau~\widehat{\htau}(p_{l,n}\otau~\hwtau-q_{m,n}\wtau~\hotau)}
 =-\dfrac{\otau~\wtau~\htau~\hwotau}{\hotau~\hwtau(p_{l,n}\otau~\hwtau-q_{m,n}\wtau~\hotau)}\notag\\
 &=\dfrac{~\otau~\wtau~\htau~}{\tau~\hotau~\hwtau}
 =\dfrac{\omega}{\,\oomega\,\womega\,},
\end{align}
which is equivalent to Equation \eqref{eqn:omega_2}.
\end{proof}

\begin{remark}
Equation \eqref{eqn:omega_1} is known as the lattice modified KdV equation \cite{NC1995:MR1329559,NQC1983:MR719638,ABS2003:MR1962121}.
\end{remark}

Consider the system \eqref{eqns:3D_uvx} expressed in the $\omega$-variable by using the equations \eqref{eqns:uvx_omega}.
It is sufficient to show that the system \eqref{eqns:3D_uvx} expressed in the $\omega$-variable can be solved by the equations \eqref{eqns:omega}.
For simplicity, we omit the description of the specific equations written in the $\omega$-variable and itemize the method below by using only words.
\begin{itemize}
\item 
It is obvious that the equations \eqref{eqn:3D_uvx_uvB}, \eqref{eqn:3D_uvx_uvC}, \eqref{eqn:3D_uvx_xvB} and \eqref{eqn:3D_uvx_xvC} written in terms of the $\omega$-variable are equivalent to Equation \eqref{eqn:omega_1}.\\[-0.7em]
\item
Firstly eliminating $\widetilde{\overline{\oomega}}$ by using Equation \eqref{eqn:omega_1}$_{l\to l+1}$,
secondly eliminating $\widetilde{\woomega}$ by using Equation \eqref{eqn:omega_1}$_{m\to m+1}$,
and finally eliminating $\woomega$ by using Equation \eqref{eqn:omega_1},
we can show that the equations \eqref{eqn:3D_uvx_u} and \eqref{eqn:3D_uvx_x} written in terms of the $\omega$-variable hold.\\[-0.7em]
\item
Eliminating $\woomega$ by using Equation \eqref{eqn:omega_1} 
and then eliminating $\hoomega$ by using \eqref{eqn:omega_2},
we can show that Equation \eqref{eqn:3D_uvx_TuvB} written in terms of the $\omega$-variable holds.\\[-0.7em]
\item
Firstly eliminating $\widehat{\overline{\oomega}}$ by using Equation \eqref{eqn:omega_2}$_{l\to l+1}$,
secondly eliminating $\hwoomega$ by using Equation \eqref{eqn:omega_2}$_{m\to m+1}$,
thirdly eliminating $\hoomega$ by using Equation \eqref{eqn:omega_2},
and finally eliminating $\woomega$ by using Equation \eqref{eqn:omega_1},
we can show that Equation \eqref{eqn:3D_uvx_TuvC} written in terms of the $\omega$-variable holds.
\end{itemize}
Therefore, we have completed the proof of Theorem \ref{theorem:tau}.

\def\cprime{$'$} \def\cprime{$'$}

\end{document}